\documentclass{article}
\usepackage[latin2]{inputenc}
\usepackage{amssymb}
\usepackage{comment}
\usepackage{makeidx}
\usepackage{gastex}
\usepackage{eepic}
\usepackage{epic}
\usepackage{graphicx}
\usepackage[usenames]{color}
\usepackage{longtable}
\usepackage[T1]{fontenc}

\title{Axiomatizing rectangular grids with no extra non-unary relations}

\newcounter{mycount}[section]

\newtheorem{theorem}[mycount]{Theorem}

\newtheorem{corollary}[mycount]{Corollary}

\newenvironment{proof}[1][]%
{\medskip{\bf Proof #1 }}%
{}
\def\qed{\hfill$\rule{2mm}{2mm}$\par\medskip}

\def\bbN{{\mathbb N}}

\def\bbZ{{\mathbb Z}}

\def\ra{\rightarrow}

\def\mythen{\Rightarrow}
\def\myiff{\iff}

\def\spec{{\rm{spec}}}

\def\SPEC{{\rm{SPEC}}}
\def\FPSPEC{{\rm{FPSPEC}}}
\def\NTISP{{\rm{NTISP}}}

\author{Eryk Kopczy\'nski}
\begin{document}

\maketitle

\begin{abstract}
We construct a formula $\phi$ which axiomatizes non-narrow rectangular grids
without using any binary relations other than the grid neighborship relations.
As a corollary, we prove that a set $A \subseteq \bbN$ is a spectrum of a formula
which has only planar models if numbers $n \in A$ can be recognized by a 
non-deterministic Turing
machine (or a one-dimensional cellular automaton) 
in time $t(n)$ and space $s(n)$, where $t(n)s(n) \leq n$ and $t(n),s(n) = \Omega(\log(n))$.
\end{abstract}

\section{Introduction}

The \emph{spectrum} of a $\phi$, denoted $\spec(\phi)$ is the set of cardinalities of
models of $\phi$. Let $\SPEC$ be the set of $A \subseteq \bbN$ such that $A$ is 
a spectrum of some formula $\phi$ is an interesting research area \cite{scholz,fiftyyears};
it is known that SPEC=NE, i.e., $A$ is a spectrum of a first order formula iff the set
of binary representations of the elements of $A$ is in the complexity class NE
\cite{fagin74,JonesS74}.
However, the characterization of spectra remains open if we require our formula, or
our models, to have additional properties. In \cite{planarspectra} we study the
complexity class $\FPSPEC$ (Forced Planar Spectra), which is the set of 
$A \subseteq \bbN$ such that there exists a formula $\phi$ such that
$\spec(\phi)=S$ and all models of $\phi$ are planar. It is shown there that
$\FPSPEC \supseteq \NTISP(n^{1-\epsilon}, \log(n))$, where $\NTISP(t(n), s(n))$
the set fo $A \subseteq \bbN$ such that there exists a non-deterministic
Turing machine which
recognizes the binary representation of $n$ in time $t(n)$ and space $s(n)$.
However, this result is not satisfying, since space $\log(n)$ is very low;
a construction of which allows more space is left as an open problem. 

In this paper we construct a formula $\phi$ over a signature consisting of only binary relations
$U,D,L,R$ (neighbors in the grid in all directions) and unary relations, and which axiomatizes rectangular grids which
are not {\it narrow}, i.e., grids of dimensions $x^* \times y^*$ where $x^* = \Omega(\log(y^*))$ and $y^* = \Omega(\log(x^*))$.
We show that it is impossible to give a similar axiomatization of rectangular grids which includes the narrow ones.
Non-narrow rectangular grids are planar graphs of bounded degree, and they can be used to simulate Turing machines, and thus we obtain the
following corollary: $\FPSPEC \supseteq \NTISP(t(n), s(n))$ for every
pair of functions $t(n), s(n)$ such that $t(n) \cdot s(n) \leq n$ and $t(n),s(n) = \Omega(\log(n))$.
In fact, we 
get a bit more -- we can actually simulate a non-deterministic one-dimensional cellular automaton (1DCA)
working in the given time and memory. While 1DCAs are less commonly taught than Turing machines,
they are simpler to define and more powerful, since they can perform computations on the whole tape at once \cite{ccotb}.

\section{Axiomatizing a rectangular grid}\label{gridax}

We obtain our goal by showing a first-order formula whose all finite models are rectangular grids.
A {\bf rectangular grid} is a relational structure $G=(V(G),L,R,U,D)$ such that $V(G) = \{0..x^*\} \times \{0..y^*\}$, and
the relations $L$, $R$, $U$, $D$ hold only in the following situations:
$L((x,y), (x-1,y))$, $R((x,y), (x+1),y)$, $U((x,y), (x,y-1))$, $D((x,y), (x,y+1))$, as long as these vertices exist.

\paragraph{Geometry}
We will use four binary relations $L$, $R$, $U$, $D$, which correspond to Left, Right, Up, Down,
respectively. We will need axioms to specify that these four relations work according to
the Euclidean square grid geometry.

\begin{itemize}
\item {\bf Partial injectivity.}
Our relations $X \in \{L, R, U, D\}$ are partial injective functions.
That is, we have an axiom $\forall x \forall y X(x,y) \wedge X(x,z) \mythen y=z$. For
$X \in \{L, R, U, D\}$, we will write $X(x)$ for the element $y$ such that $X(x,y)$ 
(if it exists). 

\item {\bf Inverses.} $\forall x \forall y R(x,y) \myiff L(y,x) \wedge U(x,y) \myiff D(y,x)$. This axiom
formalizes our interpretation of directions (that Left is inverse to Right and Up is inverse to Down).

\item {\bf Commutativity.} Let $H \in \{L,R\}$ and $V \in \{U,D\}$. Then
$\forall x \forall y \forall z H(x,y) \wedge V(x,z) \mythen \exists t H(z,t) \wedge H(y,t)$.
This axiom axiomatizes the Euclidean geometry of our grid: horizontal and vertical movements commute.
Additionally, it enforces that whenever we can go horizontally and vertically from the given $x$,
we can also combine these two movements and move diagonally.
\end{itemize}

\paragraph{Binary Counters}

We will require our grid to know its number of rows. To this end, we will introduce an extra
relation $B_V$. Intuitively, replace every vertex $v$ in the row $r=(x, R(x), R^2(x), \ldots)$, 
where $L(x)$ is not defined, with
1 if $B_V(v)$, and 0 otherwise. The axioms in this section will enforce that the obtained number
(written in the little endian binary notation) is the index of our row.

\begin{itemize}
\item {\bf Horizontal Zero.} $\forall x (\neg \exists y U(x,y)) \mythen (\neg B_V(v))$. The 
binary number encoded in the first row is zero.

\item {\bf Horizontal Increment.} 
To increment a (little endian) binary number, we change every bit which is either the leftmost one,
or such that its left neighbor changed from 1 to 0. This can be written as the following formula:
$\forall x (\exists y U(x,y)) \mythen ((B_v(x) \not \myiff B_v(U(x))) \myiff C(x)$
where $C(x) = ((\neg\exists y L(x,y)) \vee (\neg B_v(L(x)) \wedge B_v(U(L(x))))$.
%For every $x$ such that $U(x)$ is defined, $B_v(x)$ is not equal to 
%$B_v(U(x))$ iff one of the following two conditions hold: $L(x)$ is undefined or 
%$\neg B_v(L(x))$ and $B_v(U(L(x))$. Here we axiomatize the increment.

\item {\bf No Horizontal Overflow.} $\forall x (\neg \exists y R(x,y)) \mythen (\neg B_V(v))$. This axiom
makes sure that our binary counter does not overflow.
\end{itemize}

We also have analogous axioms for vertical binary counters, using an extra unary relation $B_H$, counting
from right to left, with the least significant bit on the bottom. See Figures \ref{fig}a and \ref{fig}c, where
the vertices of the grid satisfying respectively $B_V$ and $B_H$ are shown (ignore the small white circles and
dark grey boxes for now -- they will be essential for our further construction).

Let $\phi_1$ be the conjunction of all axioms above.

\begin{theorem}\label{tgeo}
If $G$ is a connected finite model of $\phi_1$ and there exists an $v \in V(G)$ and a relation $X \in \{L,R,U,D\}$ such that
$X(v)$ is not defined, then $G$ is a rectangular grid.
\end{theorem}

\begin{proof}
Take $X$ and $v$ such that $X(v)$ be not defined. Without loss of generality we can assume that $X \in \{L,R\}$ (horizontal and
vertical axioms are symmetrical). Furthermore, we can also assume that $X = L$ (since $R$ is the
inverse of $L$, if $R$ is not defined for some element, then so is $L$).

Let $v+(0,y) = D^y(v)$, where $x \geq 0$ and $y \geq 0$. From the commutativity axiom, $L(v+(0,y))$ is not defined for any
$y$. Indeed, if $L(v+(0,y))$ was defined for $y>0$, we have $L(v+(0,y))$ and $U(v+(0,y)) = v+(0,y-1)$ defined, hence 
$L(v+(0,y-1))$ is defined too.

Let $b_y = \sum 2^x [B_V(R^x(v+(0,y)))]$. From the Horizontal Increment and No Horizontal Overflow axioms, it is easy to show that $b_{y+1} = b_y+1$.
Furthermore, we have that $b_y < 2^{|V|}$. Therefore, there must exist $y$ such that $D(v+(0,y))$ is not defined. Let $v' = v+(0,y)$.
Let $x^*$ be the greatest x such that $R^x(v')$ is defined, and $y$ be the greatest $y^*$ such that $U^y(v')$ is defined. Let $G = \{0,\ldots,x^*\} \times \{0,\ldots,y^*\}$,
and for $(x,y) \in G$, let $m(x,y) = R^x(U^y(v'))$. It is straightforward that $m$ gives an isomorphism between the rectangular grid $G$ and $V$. \qed
\end{proof}

\section{Forbidding Tori}

However, rectangular grids are not the only models of $\phi_1$. Consider the torus $T = \{0, \ldots, x^*\} \times \{0, \ldots, y^*\}$, where 
$(x^*,y)$ is additionally connected (with the $R$ relation) to $(0,y$), and $(x,y^*)$ is additionally connected to $(x,0)$ (with the $U$ relation),
and we add the respective inverses to $L$ and $D$. If $B_V$ and $B_H$ are empty relations, the torus $T$ satisifes all of our axioms. Additionally,
if $G$ is a model of $\phi_1$, then the disjoint union $G \cup T$ is also a model of $\phi_1$.

To prevent this, we use the following result of Berger \cite{berger}.
\begin{theorem}\label{wangth}
There exists a finite set of Wang tiles $K = \{k_1, \ldots, k_t\}$ and relations $T_R, T_D \subseteq K \times K$ such that there exists a tiling 
$C: \bbZ \times \bbZ \ra K$ such that the following property holds:

\begin{equation}
T_R(C(x,y), C(x+1,y)) \wedge T_D(C(x,y), C(x,y+1)) \mbox{\ for each\ }x,y \in \bbZ. \label{tiling}
\end{equation}

However, no periodic tiling satisfying \ref{tiling} holds. A tiling $C: \bbZ \times \bbZ \ra K$ is {\bf periodic} iff there exists
$(x_0, y_0) \neq (0,0)$ such that $C(x,y) = C(x+x_0, y+y_0)$ for each $x,y \in \bbZ$.
\end{theorem}

The original coloring by Berger used 20426 tiles. It is sufficient to use 11 tiles \cite{jeandel}.

We add a new relation $C$ for every tile $C \in K$. We also add the following axioms:

\begin{itemize}
\item {\bf Full tiling.} $\forall v \bigvee^!_{C \in K} C(v).$ Everything needs to have a color.

\item {\bf Correct tiling.} 
For every pair of tiles $C_1, C_2 \in K$ such that $\neg T_R(C_1,C_2)$, we have $\neg \exists v C_1(v) \wedge C_2(R(v))$.
For every pair of tiles $C_1, C_2 \in K$ such that $\neg T_D(C_1,C_2)$, we have $\neg \exists v C_1(v) \wedge C_2(D(v))$.
\end{itemize}

Let $\phi_2$ be the conjuction of $\phi_1$ and the axioms above.

\begin{theorem}
If $G$ is a finite, connected model of $\phi_2$ then $G$ is a rectangular grid.
\end{theorem}

\begin{proof}
Take $v \in V$. If one of the relations $L$, $R$, $U$, $D$ is not defined for some $v \in V$, then $V$ is a rectangular grid by Theorem \ref{tgeo}. 
Otherwise, let $C(x,y)$, for $x,y \geq 0$, be the relation $C \in K$ which is satisfied by $R^x(D^y(v))$. For $x<0$ or $y<0$, replace $R^x$ by $L^{-x}$ or $D^y$ by $U^{-y}$.
According to the correct tiling axiom, the property (\ref{tiling}) holds.

Since $V$ is finite, we must have $C(x_1,y_1)$ and $C(x_2,y_2)$ refer to the same element of our structure, even though $(x_1,y_1) \neq (x_2,y_2)$. It is easy to show that
$(x_1-x_2, y_1-y_2)$ is then the period of the tiling $C$, which contradicts Theorem \ref{wangth}. \qed
\end{proof}

\begin{theorem}\label{allgrids}
There exists a formula $\phi_3$ such that the models of $\phi$, restricted to relations $L$, $R$, $U$, $D$, are precisely the rectangular grids $x^* \times y^*$ such that $y^* \leq 2^{x^*-1}$ and $x^* \leq 2^{y^*-1}$. 
\end{theorem}
\begin{proof}
By adding an axiom that there exists exactly one element $v^*$ such that $L(v^*)$ and $U(v^*)$ are not defined, we obtain a formula $\phi_3$ whose all finite models are rectangular grids.

Now, take an $x^* \times y^*$ rectangular grid G. From Theorem \ref{wangth} there exists a tiling $C: \bbZ \times \bbZ \ra K$ satisfying \ref{tiling}. Assign the relation $C(x,y)$
to each $(x,y)\in G$. If $y^* \leq 2^{x^*-1}$ and $x^* \leq 2^{y^*-1}$, we can also set $B_H(x,y)$ iff $x$-th bit of $y$ is 1, and $B_V(x,y)$ iff $y$-th bit of $x$ is 1.
Such a model will satisfy $\phi_3$. Note that if $x^* > 2^{y^*-1}$ or $y^* > 2^{x^*-1}$, the respective overflow axiom will not be satisfied. \qed
\end{proof}

The number 2 in the theorem above can be changed to an integer $b \geq 2$ by using $b$-ary counters instead of the binary ones. However:

\begin{theorem} \label{srem} \rm
There is no formula $\phi$ over a signature consisting of $L$, $R$, $U$, $D$, and possibly extra unary relations whose all models restricted to relations $L$, $R$, $U$, $D$ are
precisely all rectangular grids. Furthermore, there is no such $\phi$ such that all models of $\phi$ are rectangular grids, and 
there exists $y*$ such that for every $x*$ a rectangular model $x^* \times y^*$ of $\phi$ exists.
\end{theorem}

\begin{proof}
We will be using Hanf's locality lemma \cite{hanf}.
Let a $r$-neighborhood of the vertex $v \in V$, $N_r(v)$ be the set of all vertices whose distance from $v$ is at most $r$.
Let a {\bf $r$-type} of the vertex $v$, $\tau(v)$, be the isomorphism type of $N_r(v)$. When we restrict to models of degree bounded by $d$, 
there are only finitely many such types. Let $T_r$ be the set of all types.
Let $f_{r,M}(G): T \ra \{0..M\}$ be the function that assigns to each type $\tau \in T$ the minimum of $M$ and the number of vertices of type $\tau$ in $G$.

\begin{theorem}[Hanf's locality lemma\cite{hanf}]\label{hanf}
Let $\phi$ be a FO formula. Then there exist numbers $r$ and $M$ such that, for each graph $G=(V,E)$, $G \models \phi$ depends only on $f_{r,M}$.
\end{theorem}

Let $\phi$ be a FO formula such that all models of $\phi$ are rectangular grids. Take $r$ and $M$ from Theorem \ref{hanf}.
Let the rectangular grid $G$ be a model of $phi$, where
$V(G) = \{0,\ldots,x^*\} \times \{0,\ldots,y^*\}$. Let $\tau(x)$ be the type of column $x$, i.e., $\tau(x) = (\tau(x,0), \ldots, \tau(x,y*))$.
For sufficiently large $x^*$ there will be $x_1$ and $x_2$ such that $\tau(x_1+i) = \tau(x_2+i)$ for $i=-r, \ldots, r$ and 
such that $f_{r,M}(G)(\tau(x,y)) \geq M$ for every $x \in \{x_1, \ldots, x_2\}$. Construct a new structure $G'$ by adding a cylinder of dimensions
$(x_2-x_1) \times y*$ to $G$, i.e., 
$V(G') = V(G) \cup \{(1,x,y): x \in \{x_1, \ldots, x_2-1\}, y \in \{0, \ldots, y*\}$, $U(1,x,y) = (1,x,y-1)$, $D(1,x,y) = (1,x,y+1)$,
$R(1,x,y) = (1,x+1,y)$, $L(1,x,y) = (1,x-1,y)$, $R(1,x_2-1,y) = (1,x_1,y)$, $L(1,x_1,y) = (1,x_2-1,y)$, whenever the point on the right hand side exists,
and undefined otherwise. For every unary relation $U$ we have $U(1,x,y)$ iff $U(x,y)$. It is easy to verify that $\tau(1,x,y) = \tau(x,y)$, 
and every of these types already appeared at least $M$ times, and thus from Theorem \ref{hanf}, $G' \models \phi$. \qed
\end{proof}

\section{Forced Planar Spectra}

\begin{corollary}
Let $S \subseteq \mathbb{N}$ be a set such that there exists a non-deterministic Turing machine (or 1DCA) recognizing the set of binary representations of 
elements of $S$ in time $t(n)$ and memory $s(n)$, where $t(n) \cdot s(n) \leq n$ and $t(n),s(n) \geq \Omega(log(n))$. Then there exists a first-order formula $\phi$ such
that all models of $\phi$ are planar graphs, and the set of cardinalities of models of $\phi$ is $S$.
\end{corollary}

\begin{proof}
Let $A \subseteq \bbN$, and let $M$ be a non-deterministic Turing machine or a non-deterministic 1DCA recognizing $A$ in time $t(n)$ and space $s(n)$ such that $t(n) \cdot s(n) \leq n$. 
A non-deterministic 1DCA is $M = (\Sigma, R, F)$ where $R \subseteq \Sigma^4$, and the final symbol $F \in \Sigma$. It is defined similar to a Turing machine, 
but where computations are performed in parallel on all the tape cells: if $t(x,y)$ is the content of the tape at position $x$ and time $y$, then
the relation $R(t(x-1,y), t(x,y), t(x+1,y), t(x,y+1))$ must hold. The 1DCA accepts when it writes the symbol $F$.

Let $u(n) = n-t(n)s(n)$; without loss of
generality we can assume $u(n) < t(n)$. 
It is well known that a first order formula on a grid can be used to simulate a Turing machine (or 1DCA):
the bottom row is the initial tape, and our formula ensures that each other row above it is a correct successor of the row below it. 

Let $n \in A$. We will construct a formula $\phi$ which will have a model consisting of:
\begin{itemize}
\item A rectangular grid $G' = \{0,\ldots,s\} \times \{0,\ldots,t\}$, where $t = t(n)-1$ and $s = s(n)-1$. The structure of the grid is given by relations $L$, $R$, $U$ and $D$ 
just as in Section \ref{gridax}; we also have all the auxiliary relations required by Theorem \ref{allgrids}.
\item $u(n)$ elements which are not in the grid. The relation $P$ will hold for all the extra elements and only for them. The relation $Q$ will hold only
for the elements $(0,t-i) \in G'$ where $i \leq u(n)-1$. The relation $B$ gives a bijection between elements $x$ such that $P(x)$, and the elements $x$ such that $Q(x)$.
\item Encoding of the number $n$. 
We encode the number $n$ in the leftmost cells in the initial tape using two relations $D_n$ and $E_n$ in the following way:
$D_a(x,t)$ is the $x$-th digit of $n$, and the relation $E_a$ signifies the end of the encoding: $E_a(x,y) \mythen E_a(x+1,t) \wedge \neg D_a(x,t)$ (if $(x+1,t)$ exists). 
\item Similarly we encode the numbers $s$, $t$ and $u$.
\item An encoded run of $M$ which accepts the encoded value $n$ as the input.
\item An encoded run of an one-dimensional cellular automaton $M_2$ which verifies that the relation $n = (s+1) \times (t+1) + u$ holds for the encoded numbers.
A one-dimensional cellular automaton can add and multiply $k$-digit numbers in time $O(k)$ \cite{atrubin}, hence our space $s$ will be sufficient.
\item Our grid already has the binary representations of $s$ and $t$ computed as the relations $B_H$ and $B_V$. In the case of $B_V$ the
computed $t$ is already where we need it (we only need to define the relation $E_{t}$ in the straightforward way). In the case of $B_H$ the
computed $s$ is in the rightmost column, so we add extra wiring relations $W$ to move it to the beginning of the initial tape. In the case of $u$,
we need to compute the binary representation of the number of rows $i$ such that $Q(0,i)$; this can be computed in the same way as we have computed
the number of all rows (using the relation $B_U$ similar to $B_V$).
\end{itemize}

\def\subfig#1#2{
\begin{minipage}[b]{.5\textwidth}\includegraphics[width=\textwidth]{#1}\begin{center}#2\end{center}
\end{minipage}
}

\begin{figure}[!ht]
\subfig{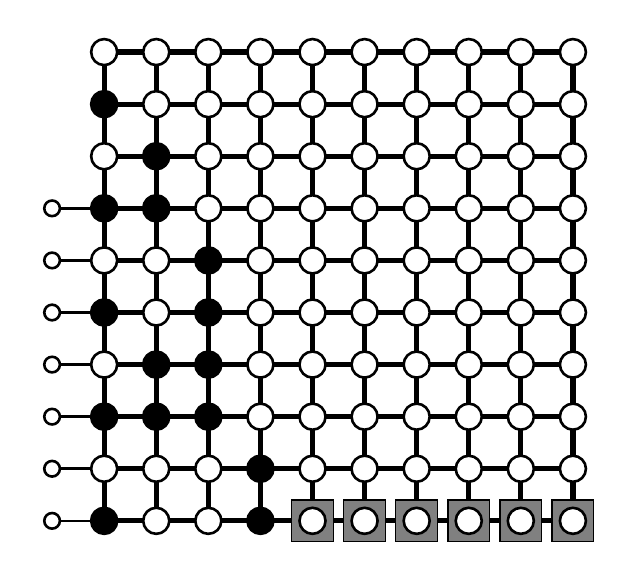}{a}\subfig{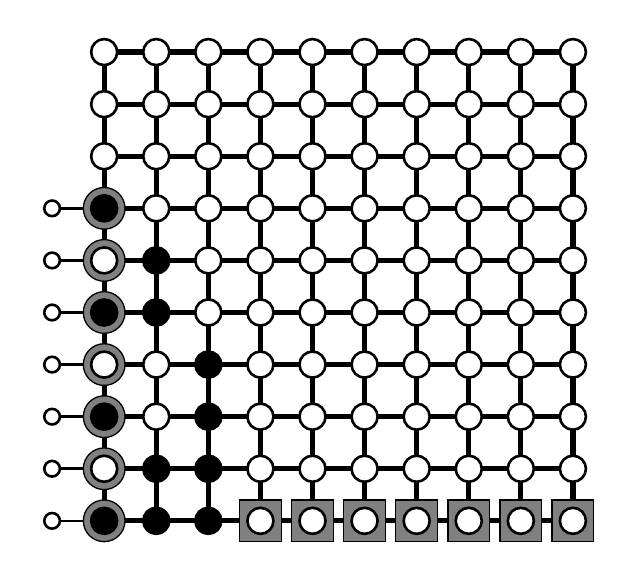}{b}

\subfig{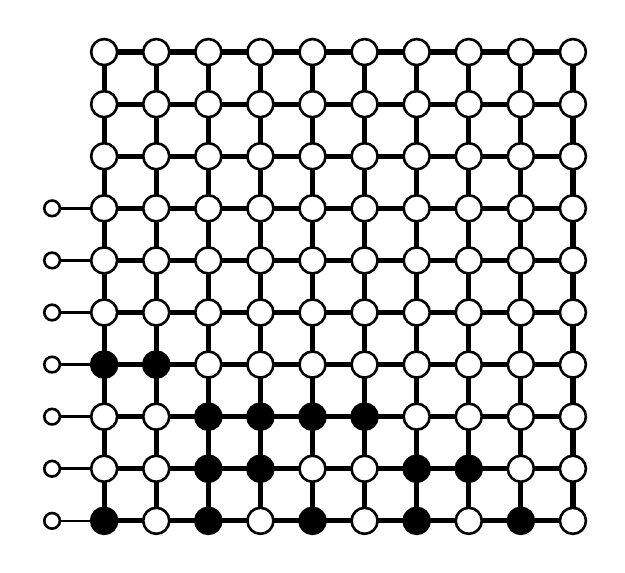}{c}\subfig{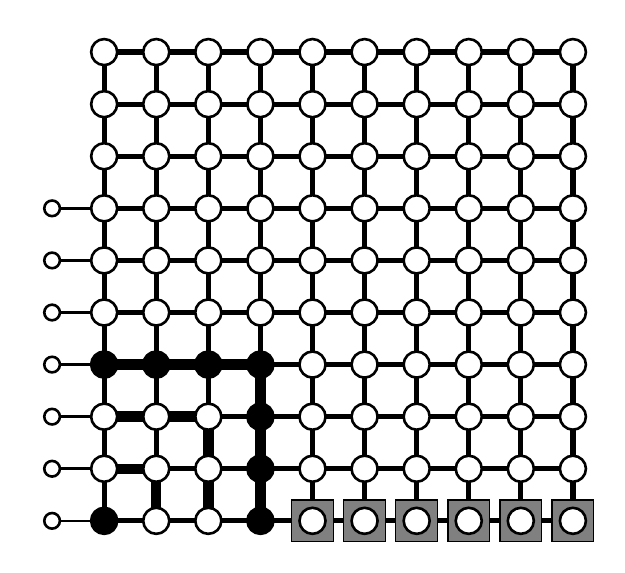}{d}
\caption{Computing the size of our model.\label{fig}}
\end{figure}

Figure \ref{fig} shows the elements of our construction. In all the pictures, the small circles are the extra elements (where $P$ holds), and the other elements are
the grid; the thin lines represent the relation $B$, the thick lines represent the relations $U$, $D$, $L$ and $R$. In \ref{fig}a the black circles
represent $B_H \equiv D_{t}$ and gray boxes represent $E_{t}$. In \ref{fig}b the gray circles represent $Q$, black circles represent $B_U$ and gray circles
represent $E_U$. In \ref{fig}c the black circles represent $B_V$, while in \ref{fig}d the extra thick lines represent $W$, black circles represent $E_{s}$,
and gray boxes represent $E_{s}$.

The formula $\phi$ will be the conjuction of the following axioms:
\begin{itemize}
\item (1) $\phi_3$, restricted to elements for which $P$ does not hold. This requires that we indeed have a rectangular grid.
\item (2) Axiomatiziations of the Turing machine $M$.
\item (3) $B$ is a bijection.
\item (4) The set of elements satisfying $Q$ has the correct shape: $Q(v) \mythen \neg P(v) \wedge (\neg\exists y L(x,y)) \wedge (\exists y D(x,y) \mythen Q(y))$, 
\item (5) Axiomatiziation of $B_U$, similar to the axiomatization of $B_V$, but where we add 1 only in the rows $y$ where $Q(0,y)$ holds.
\item (6) Axiomatiziation of the wiring $W$ moving $s$. The axioms are as follows:
$W(v,w) \wedge D_t(v) \mythen D_t(w)$; $W(v,w) \mythen W(w,v)$; every v is connected to either (a) only $R(v)$ and $L(v)$ is undefined, (b) only $L(v)$ and $R(v)$;
(c) only $L(v)$ and $D(v)$; (d) only $D(v)$ and $U(v)$; (e) only $U(v)$ and $D(v)$ is not defined; (f) nothing. Furthermore, in case (c), $L(D(v))$ must either
be also case (c) or the bottom left corner; $D_t(v) \myiff B_H(v)$ whenever $L(v)$ is undefined; and the case (a) holds whenever $L(v)$ undefined, $D_t(v)$, 
and $D(v)$ defined.
\item (7) $\forall v D_{t}(v) \myiff V_H(y)$.
\item (8) For every encoded number $a$, $E_a(v) \ra (\neg D_a(v) \wedge \exists w R(v,w) \ra E_a(w)$.
\item (9) Axiomatiziations of the automaton $M_2$.
\end{itemize}

Our model satisfies all these axioms.

On the other hand, suppose that $\phi$ has a model $G$ of size $n$. By (1) this model constists of a rectangular grid and 
a number of $u$ extra elements. By (3) and (4) the relation $Q$ is satisfied only for $u$ bottommost elements in the leftmost column.
By (5) the encoded number $u$ equals the number of these elements. By (6) and (7) the encoded numbers $s$ and $t$ equal the dimensions of
the grid. By (8) and (9) we know that the encoded number $n$ indeed equals the size of $G$. By (2) we know that $M$ accepts $n$, therefore
$n \in A$.
\end{proof}

\bibliographystyle{alpha}
\bibliography{spectra}

\end{document}